%% file: main.tex
\documentclass[11pt]{article}
\usepackage[utf8]{inputenc}

\usepackage{natbib}

\usepackage{mdframed}
\usepackage{pgfplots}

\usepackage{xcolor}

\definecolor{Gred}{RGB}{219, 50, 54}
\definecolor{Ggreen}{RGB}{60, 186, 84}
\definecolor{Gblue}{RGB}{72, 133, 237}
\definecolor{Gyellow}{RGB}{247, 178, 16}
\definecolor{ToCgreen}{RGB}{0, 128, 0}
\definecolor{myGold}{RGB}{231,141,20}
\definecolor{myBlue}{rgb}{0.19,0.41,.65}
\definecolor{myPurple}{RGB}{175,0,124}
\definecolor{niceRed}{RGB}{153,0,0}
\definecolor{niceRed}{RGB}{190,38,38}
\definecolor{blueGrotto}{HTML}{059DC0}
\definecolor{royalBlue}{HTML}{057DCD}
\definecolor{navyBlueP}{HTML}{0B579C}
\definecolor{limeGreen}{HTML}{81B622}
\definecolor{nicePink}{RGB}{247,83,148}

\usepackage{cmap}
\usepackage[T1]{fontenc}
\usepackage{bm}
\usepackage{amsmath}
\usepackage{amsfonts}
\usepackage{amssymb}
\usepackage{amsbsy}
\usepackage{amsthm}
\usepackage{bbm}
\usepackage{mathtools}

\usepackage{authblk}

\newcommand{\defeq}{\triangleq}

\usepackage{graphicx, ucs}

\usepackage{subcaption}
\usepackage{rotating}
\usepackage{float}
\usepackage{tikz}

\usepackage[linesnumbered,ruled,vlined]{algorithm2e}

\usepackage{enumitem}

\usepackage[pagebackref]{hyperref}
\hypersetup{
	colorlinks = true,
	urlcolor = {myPurple},
	linkcolor = {royalBlue},
	citecolor = {nicePink}
}

\usepackage{cleveref}

\usepackage{multirow}
\usepackage{array}

\usepackage{complexity}

\usepackage{lineno}
\let\E\relax
\let\poly\relax
\let\polylog\relax

\usepackage{chngcntr}

\counterwithin*{equation}{section}
\usepackage{chngcntr}
\usepackage{soul}
\usepackage{nicefrac}

\usepackage{fancyhdr}
\fancyheadoffset{0pt}
\def\compactify{\itemsep=0pt \topsep=0pt \partopsep=0pt \parsep=0pt}
\let\latexusecounter=\usecounter

\definecolor{myC}{rgb}{0, 255, 255}
\definecolor{myY}{rgb}{204, 204, 0}
\definecolor{myM}{rgb}{255, 0, 255}
\definecolor{secinhead}{RGB}{249,196,95}
\definecolor{lgray}{gray}{0.8}

\usepackage{thm-restate}

\newtheorem{theorem}{Theorem} 
\newtheorem*{theorem*}{Theorem} 
\newtheorem*{proposition*}{Proposition} 
\newtheorem{lemma}[theorem]{Lemma}
\newtheorem{claim}[theorem]{Claim}
\newtheorem{proposition}[theorem]{Proposition}
\newtheorem{fact}[theorem]{Fact}
\newtheorem{corollary}[theorem]{Corollary}
\newtheorem{conjecture}[theorem]{Conjecture}

\theoremstyle{definition}

\newtheorem{definition}[theorem]{Definition}
\newtheorem{remark}[theorem]{Remark}

\newcommand{\E}{\mathop{\mathbb E\/}}

\newcommand{\reg}{\mathsf{Reg}}

\newcommand{\range}[1]{[\![#1]\!]}
\newcommand{\vmu}{\vec{\mu}}
\newcommand{\kib}{\mathsf{K}}

\usepackage{xfrac}

\newcommand{\poly}{\textnormal{poly}}
\newcommand{\polylog}{\textnormal{polylog}}

\newcommand{\hatvq}{\widehat{\vec{q}}}
\newcommand{\devX}{\vec{x}^\dagger}
\newcommand{\tilvq}{\widetilde{\vec{q}}}
\newcommand{\elllog}{\ell}

\newcommand{\nats}{\mathbb N}

\newcommand{\mwu}{\mathtt{MWU}}
\newcommand{\ftpl}{\mathtt{FTPL}}
\newcommand{\omwu}{\mathtt{OMWU}}

\newcommand{\eps}{\varepsilon}

\newcommand{\calA}{\mathcal{A}}

\newcommand{\calC}{\mathcal{C}}

\newcommand{\calE}{\mathcal{E}}

\newcommand{\calG}{\mathcal{G}}
\newcommand{\calH}{\mathcal{H}}

\newcommand{\calJ}{\mathcal{J}}

\newcommand{\calO}{\mathcal{O}}

\newcommand{\calR}{\mathcal{R}}
\newcommand{\calS}{\mathcal{S}}
\newcommand{\calT}{\mathcal{T}}

\newcommand{\calX}{\mathcal{X}}

\newcommand{\calZ}{\mathcal{Z}}

\newcommand{\vq}{\vec{q}}
\newcommand{\vx}{\vec{x}}
\newcommand{\vp}{\vec{p}}
\newcommand{\va}{\vec{a}}
\newcommand{\vu}{\vec{u}}

\newcommand{\uni}{\mathsf{U}}

\newcommand{\barvmu}{\bar{\vec{\mu}}}

\let\R\relax

\newcommand{\R}{\mathbb{R}}
\newcommand{\Q}{\mathbb{Q}}

\def\<{\langle}
\def\>{\rangle}

\newcommand{\N}{\mathbb{N}}

\DeclareMathOperator*{\argmax}{argmax}

\def\vec{\bm}
\def\mat{\mathbf}

\def\dtv{D_{\mathrm{TV}}}

\makeatletter
\renewenvironment{abstract}{%
	\if@twocolumn
	\section*{\abstractname}%
	\else %% <- here I've removed \small
	\begin{center}%
		{\bfseries \large\abstractname\vspace{\z@}}%  %% <- here I've added \Large
	\end{center}%
	\quotation
	\fi}
{\if@twocolumn\else\endquotation\fi}
\makeatother

\title{On the Complexity of Computing Sparse Equilibria and \\Lower Bounds for No-Regret Learning in Games}

\author[1]{Ioannis Anagnostides}
\author[2]{Alkis Kalavasis}
\author[3]{Tuomas Sandholm}
\author[4]{Manolis Zampetakis}

\affil[1, 3]{Carnegie Mellon University}
\affil[2, 4]{Yale University}
\affil[3]{Strategy Robot, Inc.}
\affil[3]{Optimized Markets, Inc.}
\affil[3]{Strategic Machine, Inc.}
\date{\today}

\begin{document}

\maketitle

\begin{abstract}
     Characterizing the performance of no-regret dynamics in multi-player games is a foundational problem at the interface of online learning and game theory. Recent results have revealed that when all players adopt specific learning algorithms, it is possible to improve exponentially over what is predicted by the overly pessimistic no-regret framework in the traditional adversarial regime, thereby leading to faster convergence to the set of \emph{coarse correlated equilibria (CCE)}---a standard game-theoretic equilibrium concept. Yet, despite considerable recent progress, the fundamental complexity barriers for learning in normal- and extensive-form games are poorly understood.

    In this paper, we make a step towards closing this gap by first showing that---barring major complexity breakthroughs---any polynomial-time learning algorithms in extensive-form games need at least $2^{\log^{1/2 - o(1)} |\mathcal{T}|}$ iterations for the average regret to reach below even an absolute constant, where $|\mathcal{T}|$ is the number of nodes in the game. This establishes a superpolynomial separation between no-regret learning in normal- and extensive-form games, as in the former class a logarithmic number of iterations suffices to achieve constant average regret. Furthermore, our results imply that algorithms such as multiplicative weights update, as well as its \emph{optimistic} counterpart, require at least $2^{(\log \log m)^{1/2 - o(1)}}$ iterations to attain an $O(1)$-CCE in $m$-action normal-form games under any parameterization. These are the first non-trivial---and dimension-dependent---lower bounds in that setting for the most well-studied algorithms in the literature.

    From a technical standpoint, we follow a beautiful connection recently made by Foster, Golowich, and Kakade (ICML '23) between \emph{sparse} CCE and Nash equilibria in the context of Markov games. Consequently, our lower bounds rule out polynomial-time algorithms well beyond the traditional online learning framework, capturing techniques commonly used for accelerating centralized equilibrium computation.
\end{abstract}

\input{text/intro}
\input{text/prels}

\section{Lower Bounds for No-Regret Learning in Games}
\label{sec:main}

In this section, we present our main results regarding the problem of computing sparse CCE in extensive-form games, as well as the implied lower bounds for no-regret learning in games.

To do so, we build on the reduction of~\citet{foster2023hardness} targeting Markov (aka. stochastic) games. In particular, we assume that we are given as input a two-player general-sum game $\calG$ where each player has $m \in \N$ actions; that is, $|\calA_1| = |\calA_2| = m \geq 2$. We may also posit that every entry in the payoff matrices, say $\mat{M}_1, \mat{M}_2 \in \Q^{\calA_1 \times \calA_2}$, can be represented with a number of bits polynomial in $m$. Further, we assume without any loss of generality that $|\mat{M}_1[a_1, a_2]|, |\mat{M}_2[a_1, a_2]| \leq 1$, for any combination of actions $(a_1, a_2) \in \calA_1 \times \calA_2$. The key idea of the reduction is to show that a sparse CCE in a suitably constructed extensive-form game $\calT = \calT(\calG)$ (described in \Cref{def:lifted-EFG}) can be used to obtain a Nash equilibrium in the original game $\calG$; in turn, the computational hardness of Nash equilibria in two-player games~\citep{Rubinstein18:Inapproximability} will preclude polynomial-time computation of sparse CCE under a certain sparsity regime (\Cref{theorem:sparseCCE}).

In what follows, \Cref{def:lifted-EFG} formally introduces the lifted extensive-form game $\calT(\calG)$; \Cref{sec:algo} describes the process whereby a sparse CCE in $\calT$ yields a Nash equilibrium in $\calG$ (\Cref{algo:EFG}); \Cref{sec:inference} establishes the correctness of \Cref{algo:EFG}; and \Cref{sec:implic} provides the main implications for computing sparse CCE in extensive-form games, as well as no-regret learning in normal- and extensive-form games.

\subsection{The lifted extensive-form game}
\label{def:lifted-EFG}

\citet{foster2023hardness} introduce two separate reductions in order to prove hardness results in Markov games, which differ depending on whether players' policies are allowed to be Markovian or not. In extensive-form games, strategies are of course not constrained to be Markovian since they can depend arbitrarily on the information available to that player. Accordingly, we will adapt the reduction of~\citet{foster2023hardness} that targets non-Markovian policies, which in turn is based on the reduction of~\citet{Borgs10:Myth}, leading to the lifted game described in this subsection.
    
As we explained in \Cref{remark:sim}, it will be convenient for our exposition to work with extensive-form games that include simultaneous moves; again, this comes without any essential loss since simultaneous moves can always be cast as sequential moves using imperfect information, a transformation that does not qualitatively alters our results.

Now, let $\calG$ be the original two-player game in normal form. The basic idea is to construct an extensive-form game $\calT = \calT(\calG)$ consisting of $H$ repetitions of $\calG$, for a sufficiently large parameter $H \in \N$ to be specified later (\Cref{theorem:inference-EFGs}). Following the approach of~\citet{foster2023hardness}, a key ingredient is the addition of an auxiliary player, namely the \emph{Kibitzer}, which is in turn based on the hardness result of~\citet{Borgs10:Myth} pertaining the computation of Nash equilibria in repeated games. Specifically, \citet{Borgs10:Myth} reduced computing Nash equilibria in two-player normal-form games to computing Nash equilibria in three-player repeated games, thereby establishing that---the folk theorem notwithstanding---the latter problem is hard. Interestingly, this is not the case for two-player repeated games where polynomial-time algorithms do exist~\citep{Littman03:polynomial}; this suggests that proving hardness results for two-player extensive-form games could require a very different approach. So, returning to our reduction, $\calT$ here is a three-player (extensive-form) game. By convention, player $\kib \defeq 3$ will represent the Kibitzer; we often use the symbol $\kib$ instead of the index $i = 3$ for convenience in the presentation.

In each possible decision node (or simply state) $s \in \calS$ of $\calT$ each player simultaneously selects an action.\footnote{Here, we denote decision nodes with the symbol $\calS$ instead of $\calH$ as in \Cref{sec:efgs} because $\calT$ features simultaneous moves as well. We clarify that $\calS$ contains precisely the information sets of each player in the sequential representation.} Specifically, each of the first two players select actions from $\calA_1$ and $\calA_2$, respectively (where those action sets are as given in the original game $\calG$), while the action set of the Kibitzer, $\calA_\kib$, is defined as
\begin{equation*}
    \calA_\kib \defeq \{(i, a_i) : i \in [2], a_i \in \calA_i \}.
\end{equation*}

As such, each state $s \in \calS$ is in bijective correspondence with a sequence of joint actions; it is critical in this construction that each player gets to observe the other players' actions from earlier rounds, for reasons that will become clear shortly. Further, the utilities of the players are then defined as follows. For a repetition $h \in \range{H}$ and a joint action profile $(a_{1, h}, a_{2, h}, a_{\kib, h}) \in \calA_1 \times \calA_2 \times \calA_\kib$, with $a_{\kib, h} = (1, a'_{1, h})$, we define
\begin{equation*}
    u_{i, h}(a_{1, h}, a_{2, h}, a_{\kib, h}) \defeq 
    \begin{cases}
    0 & : i \neq {1}, \kib, \\
    \frac{1}{H} \left( \mat{M}_1[a_{1, h}, a_{2, h}] - \mat{M}_1[a_{1, h}' , a_{2, h}] \right) & : i = 1, \\
    \frac{1}{H} \left( \mat{M}_1[a_{1, h}', a_{2, h}] - \mat{M}_1[a_{1, h}, a_{2, h}] \right) & : i = \kib;
    \end{cases}
\end{equation*}
the utility functions are defined symmetrically when $a_{\kib, h} = (2, a'_{2, h})$. Specifically, we assume here that those rewards are given to the corresponding node in the game tree; while it is common---as we described earlier in \Cref{sec:efgs}---to assign utilities only at leaf nodes, it is clear that one can always push all the utilities in the corresponding leaf nodes without altering the equilibria of the game. Indeed, for a sequence of joint actions $(\va_1, \dots, \va_H)$, which uniquely specifies a leaf node $z \in \calZ$, the cumulative utility can be defined as $u_i : \calZ \ni z \mapsto \sum_{h=1}^H u_{i, h}(a_{1, h}, a_{2, h}, a_{\kib, h})$. We note that normalizing by the factor $H$ in $u_{i, h}(\cdot)$ above ensures that the cumulative payoffs in $\calT$ are indeed in $[-1, 1]$. We further remark that $\calT$ is a zero-sum game since $\sum_{i=1}^3 u_i(z) = 0$, for any $z \in \calZ$. Finally, it is evident that $\calT$ is indeed a perfect-recall game. 

We next state a straightforward fact, which follows directly from the definition of each utility function $u_{i, h}(\cdot)$.

\begin{lemma}
    \label{lemma:nonnegativity}
    For any repetition $h \in \range{H}$, player $i \in \range{3}$, and strategies $\vx_{-i, h} \in \bigtimes_{i' \neq i} \Delta(\calA_{i'})$, it holds that
    $\max_{a_{i, h} \in \calA_i} \E_{\va_{-i, h} \sim \vx_{-i, h}} [u_{i, h}( a_{i, h}, \va_{-i, h} )] \geq 0$.
\end{lemma}

Another simple but important observation regarding the representation of $\calT$ is the following bound on the number of nodes of $\calT$, which will be represented as $|\calT|$.

\begin{claim}
    \label{claim:nnodes}
    Let $\calG$ be a two-player $m$-action game. For the induced extensive-form game $\calT = \calT(\calG)$ it holds that $|\calT| \leq 2^{H+1} m^{3 H + 3}$.
\end{claim}

\begin{proof}
    It is clear from our construction of the extensive-form game $\calT$ that $|\calT|$ can be expressed as $1 + 2 m^3 + \dots + (2m^3)^{H} \leq 2^{H+1} m^{3 H + 3}$.
\end{proof}

In particular, the description of the extensive-form game $\calT$ is polynomial in the description of $\calG$ when $H$ is an absolute constant. In stark contrast, it is important to point out that the normal-form representation of $\calT$ is exponential even if $H = 2$. Indeed, each player would have to specify an action in each of $1 + 2m^3$ decision nodes, which leads to at least $m^{m^3}$ combinations in the normal-form representation for each player; this is why proving non-trivial hardness results for normal-form games appears to require a different approach. We also remark that the bound $m^{\Theta(H)}$ of \Cref{claim:nnodes} clearly holds after we convert $\calT$ into a sequential-move game, which suffices for our proof to carry over without simultaneous moves.

\subsection{The algorithm}
\label{sec:algo}

Based on the extensive-form game $\calT$ described in \Cref{def:lifted-EFG}, our main reduction is summarized in \Cref{algo:EFG}. Before we proceed, let us make some clarifications. First, the function $\textsc{StateToSeq}(\cdot)$ in \Cref{line:seq} takes as input a state $s_h \in \calS_h$ corresponding to the $h$th repetition, and returns the unique sequence of joint actions $(\va_1, \dots, \va_{h-1})$ that leads to that state; if $h = 1$, we can assume that it returns the empty sequence. Further, the function $\textsc{PrevStates}(\cdot)$ in \Cref{line:prevstates} takes again as input a state $s_h \in \calS_h$ and returns the unique sequence of preceding states $(s_1, \dots, s_{h-1}) \in \calS_1 \times \dots \times \calS_{h-1}$; if $h=1$, this function is again assumed to return the empty sequence. With those semantics in mind, we point out that the condition in \Cref{line:if} is activated if and only if there exists $t \in \range{T}$ such that $\vx^{(t)}_{i, s_\upsilon} [a_{i, \upsilon}] > 0$, for all $\upsilon = 1, 2, \dots h-1$; in the contrary case, the corresponding part of the tree is reached with probability $0$ under the random process of interest (as defined in the proof of \Cref{theorem:inference-EFGs} below), in which case we may set $\tilvq_{i, s_h}$ to an arbitrary distribution over $\range{T}$ (\emph{e.g.}, the uniform as in \Cref{line:arbitrary}).

\begin{algorithm}[!ht]
    \caption{Reduction for \Cref{theorem:inference-EFGs}}
    \label{algo:EFG}
    \textbf{Input}: Two-player $m$-action game $\calG$ in normal form; accuracy $\epsilon > 0$; sparsity $T \in \N$ \\
    \textbf{Output}: A $(9 \epsilon)$-Nash equilibrium of $\calG$ \\
    Construct the three-player extensive-form game $\calT(\calG)$ with $H \geq \frac{\log T}{\epsilon^2}$ (\Cref{def:lifted-EFG}) \\
    Compute a $T$-sparse $\epsilon$-CCE $\frac{1}{T} \sum_{t=1}^T \bigotimes_{i=1}^3 \vx_i^{(t)}$ (\Cref{def:sparseCCE}) in $\calT$ \label{line:CCE} \\
    \For{$h \in \range{H}$}{
    \For{$s_h \in \calS_h$} {
    $(\va_1, \dots, \va_{h-1}) \defeq \textsc{StateToSeq}(s_h)$ \label{line:seq} \\
    $(s_1, \dots, s_{h-1}) \defeq \textsc{PrevStates}(s_h)$ \label{line:prevstates} \\
    \For {$i \in \range{2}$}{
    \If{$ \sum_{t = 1}^T \exp \left( - \sum_{\upsilon=1}^{h-1} \log  \left(\frac{1}{ \vx_{i, s_{\upsilon}}^{(t)}[a_{i, \upsilon}] } \right) \right) > 0$ \label{line:if} }{
    Let
    $$ \tilvq_{i, s_h}^{(t)} \defeq \frac{ \exp \left( - \sum_{\upsilon=1}^{h-1} \log \left( \frac{1}{\vx_{i, s_\upsilon}^{(t)}[a_{i, \upsilon}]} \right) \right) }{ \sum_{t' = 1}^T \exp \left( - \sum_{\upsilon=1}^{h-1} \log  \left(\frac{1}{ \vx_{i, s_{\upsilon}}^{(t')}[a_{i, \upsilon}] } \right) \right) } \quad \forall t \in \range{T} $$ \label{line:EW}\\
    }
    \Else{
    $\tilvq_{i, s_h} \defeq \uni(\range{T}) $\label{line:arbitrary} \\
    }
    $\hatvq_{i, s_h} \defeq \E_{t \sim \tilvq_{i, s_h}} [ \vx^{(t)}_{i, s_h}] \in \Delta(\calA_i)$ \label{line:qhat} \\
    }
    \If{$(\hatvq_{1, s_h}, \hatvq_{2, s_h}) \in \Delta(\calA_1) \times \Delta(\calA_2)$ is a $(9 \epsilon)$-Nash equilibrium of $\calG$}{\textbf{return} $(\hatvq_{1, s_h}, \hatvq_{2, s_h}) \in \Delta(\calA_1) \times \Delta(\calA_2)$} \label{line:return}
    }
    }
    \textbf{return} \textsf{FAIL}
\end{algorithm}

It is evident that as long as $T = \poly(|\calT|)$, all steps in \Cref{algo:EFG} can be implemented in time polynomial in the description of $\calT$, with the exception of \Cref{line:CCE}, which of course depends on the underlying algorithm used to compute a sparse CCE. Along with \Cref{claim:nnodes}, we arrive at the following conclusion.

\begin{proposition}
    \label{prop:runtime}
    Let $\mathfrak{A}$ be an algorithm that takes as input an extensive-form game $\calT$ and computes a $T$-sparse $\epsilon$-CCE of $\calT$ in time at most $Q(|\calT|, T, 1/\epsilon)$. Then, \Cref{algo:EFG} instantiated with $\mathfrak{A}$ in \Cref{line:CCE} runs in time at most $Q(|\calT|, T, 1/\epsilon) +  T  m^{\Theta(H)}$.
\end{proposition}

\begin{remark}[Bit complexity of exponential weights]
    \Cref{line:EW} of \Cref{algo:EFG} updates $\tilvq_{i, s_h}^{(t)}$ using exponential weights (in accordance with the aggregation algorithm~\eqref{eq:Vovk}), which could result in $\tilvq_{i, s_h}^{(t)}$ taking irrational values. This can be addressed by simply truncating those values to a sufficiently large polynomial number of bits, in which case the proof of \Cref{theorem:inference-EFGs} readily carries over. For simplicity, we assume in our analysis that $\tilvq_{i, s_h}^{(t)}$ is updated per \Cref{line:EW}, without taking into account the numerical imprecision.
\end{remark}

It is worth commenting here on a couple of differences with~\citep[Algorithm 2]{foster2023hardness}. First, \citet{foster2023hardness} had to encode the joint action profile through the reward, so as to ensure that each player has observed the prior sequence of joint actions. This is not necessary in our setting since, by construction, the states of the extensive-form game $\calT$  encode that information. Further, their algorithm is randomized since---among other steps---they sample a randomized trajectory under a certain random process. To obtain a deterministic algorithm, we instead essentially search over all possible trajectories---all states of the extensive-form game $\calT$---for a Nash equilibrium, which we can afford in our setting.

\subsection{From sparse CCE in $\calT$ to Nash equilibria in $\calG$}
\label{sec:inference}

We are now ready to proceed with the key proof of this section, which establishes the correctness of \Cref{algo:EFG}.

\begin{theorem}
    \label{theorem:inference-EFGs}
    When $H \geq \frac{ \log T}{\epsilon^2}$, \Cref{algo:EFG} returns a $(9 \epsilon)$-Nash equilibrium in the two-player $m$-action game $\calG$.
\end{theorem}

\begin{proof}
    We consider a sequence of joint strategies $(\vx^{(1)}_1, \vx^{(1)}_2, \vx^{(1)}_3), \dots, (\vx^{(T)}_1, \vx^{(T)}_2, \vx^{(T)}_3) \in \bigtimes_{i=1}^3 \calX_i$ in the extensive-form game $\calT$ with the property that $\barvmu \defeq \frac{1}{T} \sum_{t=1}^T \bigotimes_{i=1}^3 \vx_i^{(t)}$ is an $\epsilon$-CCE, and by construction $T$-sparse per \Cref{def:sparseCCE}.
    
    Let us fix a player $i \in \range{3}$. For each state $s \in \calS$, we define $\tilvq_{i, s} \in \Delta(\range{T})$ per \Cref{line:EW}; $\hatvq_{i, s} \in \Delta(\calA_i)$ per \Cref{line:qhat}; and the deviation strategy $\devX_{i} \in \calX_i$ so that for each state $s \in \calS$ it holds that $\devX_{i, s} \defeq \argmax_{a_{i, s} \in \calA_i} \E_{\va_{-i, s} \sim \hatvq_{-i, s}} [u_{i, h}(a_{i, s}, \va_{-i, s})]$. We will now make use of \Cref{prop:aggregation} regarding the aggregation algorithm~\eqref{eq:Vovk} under a certain random process to be described shortly. In particular, under a different player $i' \neq i$, to relate our problem with the setup of online density estimation introduced earlier, we  make the following correspondence:
    \begin{itemize}
        \item the context space $\calO$ corresponds to the set of all possible states or decision nodes $\calS$ of the extensive-form game $\calT$;
        \item the set of experts $\calE$ coincides with the set $ \{ \vx_{i', s}^{(1)}, \dots, \vx_{i', s}^{(T)} \} $, with outcome space $\calA_{i'}$; and
        \item the time index $h \in \range{H}$ in the context of online density estimation will now (fittingly) correspond to the repetition $h \in \range{H}$.
    \end{itemize}
    We note that, by construction of the extensive-form game $\calT$, player $i$ observes the underlying state $s_h$ at each repetition $h \in \range{H}$, which fully specifies the sequence of joint actions leading up to that state. As a result, under a given random sequence of states $(s_1, \dots, s_H) \in \calS^H$, we can apply the aggregation algorithm~\eqref{eq:Vovk} with the aforementioned parameterization to obtain an estimate $\hatvq_{i', s_h} \in \Delta(\calA_{i'})$ for all repetitions $h \in \range{H}$ and $i' \neq i$.
    Below, we overload the notation by letting $\hatvq_{i', h} \defeq \hatvq_{i', s_h}$ and $\tilvq_{i', h} \defeq \tilvq_{i', s_h}$ so as to be consistent with the notation of \Cref{sec:ode}.
    Namely, we have that 
    \begin{equation*}
        \hatvq_{i', h} \defeq \E_{t \sim \tilvq_{i', h}}[\vec{x}_{i', s_h}^{(t)}], \textrm{ where~ } \tilvq_{i', h}^{(t)} \defeq \frac{ \exp \left( - \sum_{\upsilon=1}^{h-1} \elllog_{\upsilon}(\vx_{i', s_\upsilon}^{(t)}) \right) }{ \sum_{t' = 1}^T \exp \left( - \sum_{\upsilon=1}^{h-1} \elllog_{\upsilon}(\vx_{i', s_\upsilon}^{(t')}) \right) } \quad \forall t \in \range{T}.
    \end{equation*}
    Given that the sequence of states $(s_1, \dots, s_H)$ is produced by a certain random process (described next), $\hatvq_{i', h}$ and $\tilvq_{i', h}$ are random variables. We also recall that $\elllog_\upsilon( \vx^{(t)}_{i', s_\upsilon}) = \log \frac{1}{ \vx^{(t)}_{i', s_\upsilon}[a_{i', \upsilon}] }$. Accordingly, the deviation $\devX_{i, h} \in \Delta(\calA_i) $ for player $i \in \range{3}$ is defined as follows:
       \begin{equation}
            \label{eq:devq}
           \devX_{i, h} \defeq \argmax_{a_{i, h} \in \calA_i} \E_{\vec{a}_{-i, h} \sim  \hatvq_{-i, h}} [u_{i, h}(a_{i, h}, \va_{-i, h})].
       \end{equation}
    We next argue about the deviation benefit of player $i$ under the deviation strategy described above. In particular, we are interested in the payoff player $i$ obtains under the random process wherein we first draw an index $t^\star$ uniformly at random from the set $\range{T}$, player $i$ plays according to the deviation strategy $\devX_i$, while the rest of the players play according to $\vx^{(t^\star)}_{-i}$. By definition of this random process, realizability is satisfied: the observed distribution of outcomes obeys the law induced by $\vx^{(t^\star)}_{i'} : \calS \to \Delta(\calA_{i'})$, conditioned on the time index $t^\star$. As a result, \Cref{prop:aggregation} implies that for each player $i \in \range{3}$ and $i' \neq i$ it holds that
    \begin{equation}
        \label{eq:TV-close}
        \E_{ \devX_i \times \vx_{-i}^{(t^\star)}} \left[ \sum_{h = 1}^H \dtv \left( \hatvq_{i', h}, \vx^{(t^\star)}_{i', s_h} \right) \right] \leq \sqrt{H \log T},
    \end{equation}
    where the expectation and the sequence of states $(s_1, \dots, s_H) \in \calS^H$ is taken with respect to the random process described above. As a result, we have that
    \begin{align}
        u_i(\devX_i, \barvmu_{-i}) &= \E_{t^\star \sim \uni(\range{T})} [ u_i(\devX_i, \vx_{-i}^{(t^\star)}) ] = \E_{t^\star \sim \uni(\range{T})} \E_{ \devX_i, \vx_{-i}^{(t^\star)}} \sum_{h=1}^H \E_{\va_{-i, h} \sim \vx^{(t^\star)}_{-i, s_h}} [u_{i, h}( \devX_{i, h}, \va_{-i, h} )] \notag \\
        &\geq - 2 \sqrt{\frac{\log T}{H}} + \E_{t^\star \sim \uni(\range{T})} \E_{ \devX_i, \vx_{-i}^{(t^\star)}} \sum_{h=1}^H \E_{\va_{-i, h} \sim \hatvq_{-i, h}  } [u_{i, h}( \devX_{i, h}, \va_{-i, h} )] \label{align:TV-sum} \\
        &\geq - 2 \sqrt{\frac{\log T}{H}} + \E_{t^\star \sim \uni(\range{T}) } \E_{ \devX_i, \vx_{-i}^{(t^\star)}} \sum_{h=1}^H \max_{a_{i, h} \in \calA_i} \E_{\va_{-i, h} \sim \hatvq_{-i, h}} [u_{i, h}( a_{i, h}, \va_{-i, h} )] \label{align:opt-dev} \\
        &\geq- 2 \sqrt{\frac{\log T}{H}},\label{align:product-nonnegative}
    \end{align}
    where \eqref{align:TV-sum} uses~\eqref{eq:TV-close} along with the fact that 
    \begin{align*}
        \E_{\va_{-i, h} \sim \vx^{(t^\star)}_{-i, s_h}} [u_{i, h}( \devX_{i, h}, \va_{-i, h} )] &\geq \E_{\va_{-i, h} \sim \hatvq_{-i, h}  } [u_{i, h}( \devX_{i, h}, \va_{-i, h} )] - \frac{1}{H} \dtv\left( \vx^{(t^\star)}_{-i, s_h}, \bigtimes_{i' \neq i} \hatvq_{i', h}  \right) \\
         &\geq \E_{\va_{-i, h} \sim \hatvq_{-i, h}  } [u_{i, h}( \devX_{i, h}, \va_{-i, h} )] - \frac{1}{H} \sum_{i' \neq i} \dtv\left( \vx^{(t^\star)}_{i', s_h},  \hatvq_{i', h}  \right), 
    \end{align*}
    since $|u_{i,h}(\cdot, \cdot)| \leq \frac{1}{H}$ (by construction) and the total variation distance between two product distributions is bounded by the sum of the total variation of the individual components~\citep{Hoeffding58:Distinguishability}; \eqref{align:opt-dev} follows from the definition of $\devX_{i, h}$ in~\eqref{eq:devq}, which in particular implies that 
    $$\E_{\va_{-i, h} \sim \hatvq_{-i, h}  } [u_{i, h}( \devX_{i, h}, \va_{-i, h} )] = \max_{a_{i, h} \in \calA_i} \E_{\va_{-i, h} \sim \hatvq_{-i, h}} [u_{i, h}( a_{i, h}, \va_{-i, h} )];$$ and~\eqref{align:product-nonnegative} follows from the fact that $\max_{a_{i, h} \in \calA_i} \E_{\va_{-i, h} \sim \hatvq_{-i, h}} [u_{i, h}( a_{i, h}, \va_{-i, h} )] \geq 0$ (\Cref{lemma:nonnegativity}). Further, since $\barvmu$ is assumed to be an $\epsilon$-CCE, \eqref{align:product-nonnegative} implies that for each player $i \in \range{3}$,
    \begin{equation}
        \label{eq:aux}
        u_i(\barvmu) \geq - 2 \sqrt{\frac{\log T}{H}} - \epsilon. 
    \end{equation}
    
    Given that $\calT$ is zero-sum, we also have that $\sum_{i=1}^3 u_i(\barvmu) = 0$; by~\eqref{eq:aux}, this in turn implies that $u_\kib(\barvmu) = - u_1(\barvmu) - u_2(\barvmu) \leq 4 \sqrt{\frac{\log T}{H}} + 2 \epsilon$. We next focus on analyzing the deviation benefit of the Kibitzer. We let
    \begin{equation*}
         \delta(\hatvq_{1,h}, \hatvq_{2,h}) \defeq \max \left\{ \max_{a'_{1, h} \in \calA_{1}}  \left\{\mat{M}_{1}[a_{1, h}', \hatvq_{2, h}] - \mat{M}_{1}[\hatvq_{-\kib, h}]\right\}, \max_{a'_{2, h} \in \calA_{2}} \left\{\mat{M}_{2}[\hatvq_{1, h}, a'_{2, h}] - \mat{M}_{2}[\hatvq_{-\kib, h}]\right\}  \right\}.
    \end{equation*}
    By~\eqref{align:opt-dev}, we have that
    \begin{align*}
         u_\kib(\devX_\kib, \barvmu_{-\kib}) \geq \frac{1}{H} \E_{t^\star \sim \uni(\range{T})} \E_{ \devX_\kib, \vx_{-\kib}^{(t^\star)}} \sum_{h=1}^H \delta(\hatvq_{1,h}, \hatvq_{2,h}) - 2 \sqrt{\frac{\log T}{H}}.
    \end{align*}
    Since $\barvmu$ is an $\epsilon$-CCE, we also know that $u_\kib(\devX_\kib, \barvmu_{-\kib}) \leq u_\kib(\barvmu) + \epsilon \leq 3 \epsilon + 4 \sqrt{\frac{\log T}{H}}$, which in turn implies that
    \begin{equation*}
        \E_{t^\star \sim \uni(\range{T})} \E_{ \devX_\kib, \vx_{-\kib}^{(t^\star)}} \frac{1}{H} \sum_{h=1}^H \delta(\hatvq_{1,h}, \hatvq_{2,h}) \leq 9 \epsilon,
    \end{equation*}
    where we used that $H \geq \frac{\log T}{\epsilon^2}$. Finally, given that $\delta(\hatvq_{1,h}, \hatvq_{2,h}) \geq 0$ (\Cref{lemma:nonnegativity}), we conclude that there exists some repetition $h \in \range{H}$ and a state $s \in \calS_h$ such that $\delta(\hatvq_{1,s_h}, \hatvq_{2,s_h}) \leq 9 \epsilon$. That is, $(\hatvq_{1, s_h}, \hatvq_{2, s_h})$ is a $(9\epsilon)$-Nash equilibrium, concluding the proof.
\end{proof}

\begin{remark}
    \label{remark:non-uniform}
    It is direct to see that the proof of \Cref{theorem:inference-EFGs} can be extended under a more general notion of sparse CCE, wherein $\barvmu$ is not necessarily a uniform mixture of product distributions; this notion of CCE is naturally associated with a weighted generalization of regret. This observation is important since in practice taking a non-uniform average can lead to significant gains in performance~\citep{Brown19:Solving}.
\end{remark}

\subsection{Implications}
\label{sec:implic}

Having established \Cref{theorem:inference-EFGs}, we are now ready to prove that computing approximate CCE in a certain regime of sparsity is hard, at least under some well-established complexity assumptions. In particular, we will leverage the hardness result of~\citet{Rubinstein18:Inapproximability} (\Cref{theorem:rubin}), which rests on the so-called \emph{exponential-time hypothesis (ETH)} for $\PPAD$~\citep{Babichenko16:Can}. That hypothesis pertains the complexity of solving $\textsc{EndOfALine}$, the prototypical $\PPAD$-complete problem~\citep{Papadimitriou94:On}. 

\begin{conjecture}[\citep{Babichenko16:Can}]
    \label{conj:ETH}
    Solving $\textsc{EndOfALine}$ on $m$-bit circuits with $\Tilde{O}(m)$ gates requires time $2^{\Tilde{\Omega}(m)}$.
\end{conjecture}

Assuming that this conjecture holds, \citet{Rubinstein18:Inapproximability} proved that the quasipolynomial algorithm of~\citet{Lipton03:Playing} is essentially optimal.

\begin{theorem}[\citep{Rubinstein18:Inapproximability}]
    \label{theorem:rubin}
    Assuming \Cref{conj:ETH}, there is an absolute constant $\epsilon_0 > 0$ such that finding an $\epsilon_0$-Nash equilibrium in two-player $m$-action games requires time $m^{\log_2^{1 - o(1)}m}$.
\end{theorem}

We now use \Cref{theorem:inference-EFGs} to prove the following hardness result. Below, we recall that we use the notation $|\calT|$ to represent the number of nodes in the extensive-form game $\calT$.

\sparsecce*

\begin{proof}
    Suppose that algorithm $\mathfrak{A}$ implementing \Cref{line:CCE} of \Cref{algo:EFG} runs in time polynomial in the description of $\calT$ and computes a $T$-sparse $(\epsilon_0/9)$-CCE of $\calT$, where $T = 2^{\log_2^{\gamma} |\calT|}$ and $\gamma < \frac{1}{2}$. Then, by \Cref{prop:runtime}, \Cref{algo:EFG} runs in time $m^{\Theta(H)}$. As a result, it follows from \Cref{theorem:inference-EFGs,theorem:rubin} that for $H \geq \frac{81 \log T}{\epsilon_0^2}$, it must hold that $m^{\Theta(H)} \geq m^{\log_2^{1 - o(1)}m}$. Further, \Cref{claim:nnodes} implies that $T = 2^{\log_2^{\gamma} |\calT| } \leq 2^{4 H^\gamma \log_2^\gamma m }$. As a result, for a sufficiently large $H = O(\log_2^{\frac{\gamma}{1 - \gamma}} m)$ it follows that $H \geq \frac{81 \log T}{\epsilon_0^2}$, which in turn implies that $H = \Omega( \log_2^{1 - o(1)} m)$. As a result, we conclude that $\gamma \geq \frac{1}{2} - o(1)$, leading to the desired conclusion.
\end{proof}

A number of remarks regarding \Cref{theorem:sparseCCE} are in order. First, \Cref{theorem:sparseCCE} establishes a stark separation between extensive- and normal-form games. Indeed, as we have seen, in normal-form games there are polynomial-time algorithms, such as multiplicative weights update, that can compute an $O(\log m)$-sparse $O(1)$-CCE in polynomial time. In contrast, \Cref{theorem:sparseCCE} precludes even computing a $\polylog |\calT|$-sparse $O(1)$-CCE, unless the quasipolynomial-time algorithm of~\citet{Lipton03:Playing} can be improved. It is also worth pointing out here that it is natural to expect that analogous lower bounds to \Cref{theorem:sparseCCE} could be established under the more plausible conjecture $\P \neq \PPAD$ (instead of \Cref{conj:ETH}). Yet, that seems to require a very different approach. Indeed, for the $\PPAD$-hardness of $\epsilon$-Nash equilibria in two-player games to kick in, one must take $\epsilon$ to be inversely polynomial to $m$~\citep{Chen09:Settling}. In that case, the description of $\calT$ becomes immediately $m^{\Omega(\poly(m))}$, which renders reductions analogous to \Cref{theorem:inference-EFGs} of little use. It is thus crucial for our approach to take $\epsilon > 0$ to be an (absolute) constant. Finally, we point out that \Cref{theorem:sparseCCE} still applies by taking $T = 2^{C \log_2^{1/2 - o(1) } |\calT|}$, for any absolute constant $C > 0$.

To better contextualize \Cref{theorem:sparseCCE}, we remark that there are algorithms running in time polynomial in $\calT$ that can compute an $O(k)$-sparse $O(1)$-CCE in every extensive-form game $\calT$  with $k \leq \max_{1 \leq i \leq n} \log \left( \prod_{j \in \calJ_i} |\calA_j| \right) = \max_{1 \leq i \leq n} \sum_{j \in \calJ_i}\log |\calA_j|$; that is, $k$ is at most nearly-linear in $|\calT|$, and it can be much smaller depending on the information structure of $\calT$. This is a direct consequence of the fact that algorithms such as multiplicative weights update can be implemented in polynomial time in extensive-form games~\citep{Farina22:Kernelized}, thereby implying that the regret of each player $i$ will be bounded as $O(\sqrt{ T \log m_i})$, where $m_i$ represents the number of actions in the induced normal-form game: $m_i \defeq \prod_{j \in \calJ_i} |\calA_j|$. As a result, we see that there is a gap between our lower bound (\Cref{theorem:sparseCCE}) and the aforementioned best-known upper bound, which essentially amounts to improving the exponent $\frac{1}{2}$ in the term $2^{\log_2^{1/2 - o(1)} |\calT|} $ all the way up to $1$.

\paragraph{Lower bounds for no-regret learning in extensive-form games} Relatedly, we next proceed by pointing out some important implications of \Cref{theorem:sparseCCE} for bounding the regret incurred by no-regret learning algorithms in extensive-form games. In particular, a number no-regret learning algorithms have been designed with iteration complexity polynomial in the description of the extensive-form game (\Cref{sec:related}). For example, one such broad class derives from the paradigm of \emph{follow the perturbed leader} ($\ftpl$)~\citep{Abernethy16:Perturbation,Kalai05:Efficient,Hazan16:Introduction}; indeed, $\ftpl$ can be implemented efficiently under a linear optimization oracle, which can be in turn implemented in polynomial time in extensive-form games. \Cref{theorem:sparseCCE} circumscribes the performance of any of those algorithms---and combinations thereof---when employed simultaneously by all players.

\efgsnoregret*

In words, obtaining an average regret below an absolute constant requires at least $2^{\log_2^{1/2 - o(1)} |\calT|}$ repetitions, which again stands in stark contrast to the performance of learning algorithms in normal-form games. \Cref{cor:noregret-efgs} is an immediate consequence of \Cref{theorem:sparseCCE}, along with the folklore fact that the average product distribution after $T$ repetitions---by definition $T$-sparse---constitutes a $\frac{1}{T} \max_{1 \leq i \leq n} \reg_i^T$-CCE (\Cref{prop:folklore}). We clarify that each regret $\reg_i^T$ can be indeed computed in time $\poly(|\calT|, T)$ since it amounts to computing a best response, in turn implying that one can determine for each repetition whether $\frac{1}{t} \max_{1 \leq i \leq n} \reg_i^t \leq \epsilon$.

A noteworthy feature of \Cref{cor:noregret-efgs} is that it applies even if players are following different no-regret algorithms, and the updates are not simultaneous. It is especially worth stressing that last point. One popular way of improving the performance of no-regret algorithms in games consists of \emph{alternation}~\citep{Tammelin15:Solving,Wibisono22:Alternating}, whereby players are updating their strategies in an alternating fashion. Alternation is, of course, not a legitimate choice within the framework of online learning as it trivializes the problem for the player who gets to play last; for example, that player could always just best respond, which would accumulate at most $0$ regret for that player. Nevertheless, \Cref{cor:noregret-efgs} still holds even under learning dynamics that are beyond the framework of online learning. Furthermore, unlike the paper of~\citet{Daskalakis15:Near}, \Cref{cor:noregret-efgs} does not limit the amount of memory players use, as long as the running time stays polynomial. 

\paragraph{Lower bounds for (optimistic) MWU} Beyond extensive-form games, our results also turn out to have implications for the performance of certain no-regret learning algorithms in normal-form games. In particular, one can always cast an extensive-form game in normal form, and then use a no-regret learning algorithm on the induced normal-form game. In general, such an approach is not interesting algorithmically since the iteration complexity would be exponential, thereby rendering computational lower bounds of little use. However, it turns out there are certain algorithms for which each iteration can be indeed implemented in polynomial time, even though they operate over the---typically exponentially large---normal-form representation. This can be accomplished by leveraging the underlying structure of the extensive-form representation, and it is akin to the kernel trick~\citep{Takimoto03:Path,Beaglehole23:Sampling}. Perhaps most notably, such is the case for the celebrated multiplicative weights updates ($\mwu$), as well as its \emph{optimistic} counterpart ($\omwu$)~\citep{Farina22:Kernelized} (see also~\citep{Bai22:Efficient,Chhablani23:Multiplicative,Peng23:Fast}).

In particular, let us recall ($\mathtt{O}$)$\mwu$ in the context of learning in games. In the vanilla $\mwu$ algorithm each player $i \in \range{n}$ updates its strategy for $t = 1, 2, \dots$ as follows.

\begin{equation}
    \tag{MWU}
    \vx^{(t+1)}_i[a_i] = \frac{ \vx_i^{(t)}[a_i] \exp\left( \eta_i \vu_i^{(t)}[a_i] \right) }{\sum_{a_i' \in \calA_i} \vx_i^{(t)}[a_i'] \exp \left( \eta_i \vu_i^{(t)}[a_i'] \right) }, \quad \forall a_i \in \calA_i.
\end{equation}
Here, $\vu_i^{(t)}[a_i] \defeq \E_{\va_{-i} \sim  \vx_{-i}^{(t)} } [u_i(a_i, \va_{-i}) ]$; $\vx_i^{(1)} \defeq ( \sfrac{1}{|\calA_i|}, \dots, \sfrac{1}{|\calA_i|})$; and $\eta_i > 0$ is the learning rate. Beyond the uniform distribution, one can also initialize $\mwu$ to any point in the (relative) interior of the simplex.\footnote{We should note that the analysis of~\citet{Farina22:Kernelized} pertaining the iteration complexity of $\mwu$ in extensive-form games considers the uniform initialization, but can be directly generalized.} Similarly, $\omwu$~\citep{Rakhlin13:Online} is defined via the following update rule.
\begin{equation}
    \tag{OMWU}
    \vx^{(t+1)}_i[a_i] = \frac{ \vx_i^{(t)}[a_i] \exp\left( \eta_i (2 \vu_i^{(t)}[a_i] - \vu_i^{(t-1)}[a_i] ) \right) }{\sum_{a_i' \in \calA_i} \vx_i^{(t)}[a_i'] \exp \left( \eta_i (2 \vu_i^{(t)}[a_i'] - \vu_i^{(t-1)}[a_i']) \right) }, \quad \forall a_i \in \calA_i.
\end{equation}

$\omwu$ can be seen as an variant of $\mwu$ in which a prediction term is incorporated into the update rule. In the definitions above, the player's strategies could take irrational values due to the exponential function, but this can be readily addressed by truncating to a sufficiently large number of bits, an operation that does not essentially alter any of the results.

In general, it is evident that ($\mathtt{O}$)$\mwu$ requires time $\Omega(|\calA_i|)$ in each iteration, which is typically exponential in the context of extensive-form games. However, it turns out that each iteration of ($\mathtt{O}$)$\mwu$ can be implicitly performed in time $\poly(|\calT|)$~\citep{Farina22:Kernelized}; the analysis of~\citet{Farina22:Kernelized} does not account for the numerical imprecision resulting from the exponential function, but their argument readily carries over by truncating to a sufficiently large number of bits. This leads to the following lower bound for the regret accumulated by such algorithms. Below, we tacitly assume that each learning rate is given by an efficiently computable function.

\begin{corollary}
    \label{cor:MWU}
    Consider the class of $n$-player normal-form games where each player has $m'$ actions. If each player $i \in \range{n}$ follows $\mwu$ or $\omwu$ with any learning rate $\eta_i = \eta_i(n, \log m', T)$ and incurs regret $ \reg_i^T$ after $T$ repetitions, there is a game $\calG$ and an absolute constant $\epsilon > 0$ such that at least $T \geq 2^{(\log_2 \log_2 m')^{1/2 - o(1)}}$ repetitions of the game are needed so that $\frac{1}{T} \max_{1 \leq i \leq n} \reg_i^T \leq \epsilon$, unless ETH for $\PPAD$ (\Cref{conj:ETH}) fails.
\end{corollary}

In proof, the extensive-form game $\calT(\calG)$ described in \Cref{def:lifted-EFG} can be cast as a $3$-player $m'$-action normal-form game $\calG'$, where $m' = (2m)^{m^{\Theta(H)}}$. In this normal-form game $\calG'$, each player can implement $\mwu$ or $\omwu$ with per-iteration complexity polynomial in $|\calT|$~\citep{Farina22:Kernelized}, along with a representation of the iterates in $\poly(|\calT|, T)$ space. As a result, \Cref{cor:noregret-efgs} implies that, for any constant $C > 0$, at least $T \geq 2^{C \log_2^{1/2 - o(1)} |\calT| }$ repetitions are needed so that $\frac{1}{T} \max_{1 \leq i \leq n} \reg_i^T \leq \epsilon$ for each game $\calT$. The statement of \Cref{cor:MWU} thus follows from the fact that $\log_2  |\calT| \geq \Omega( \log_2 \log_2 m')$.

Let us point out another way to express the conclusion of \Cref{cor:MWU}. Suppose that the regret of each player, who has $m'$ available actions, can be bounded as $\reg_i^T \leq R(m') T^{1 - \gamma} $, for some constant $\gamma \in (0,1]$; this is the canonical form regret bounds assume. \Cref{cor:MWU} then implies that $R(m') \geq 2^{\gamma (\log_2 \log_2 m')^{1/2 - o(1)}}$.

We suspect that, at least under a specific parameterization, there should be a more elementary way of proving unconditional dimension-dependent lower bounds when multiple players follow algorithms such as $\mwu$. The main advantage of our approach is that it applies to any parameterization (potentially game-specific) and a broad class of algorithms; as concrete examples, our approach is robust to considering alternating instead of simultaneous dynamics, different players following different variants of $\mwu$, as well as using more general prediction mechanisms within the paradigm of optimistic $\mwu$~\citep{Syrgkanis15:Fast}. Beyond $\mwu$-type update rules, we suspect that \Cref{cor:MWU} applies more broadly to any member of follow the perturbed leader ($\ftpl$).

%\paragraph{Sublinear algorithms} In many realistic applications, even traversing the entire game tree is prohibitive~\citep{McAleer23:ESCHER,Steinberger20:DREAM}. Instead, one has to resort to \emph{sublinear} learning algorithms, whose iteration complexity satisfies is much smaller than the size of the game tree. It should come to no surprise that stronger lower bounds can be attained under such algorithms. Specifically, let us make here the stronger assumption that the iteration complexity is polynomial in the depth and the maximum number of actions over each information set.

\input{text/conclusion}

\section*{Acknowledgments}

We are grateful to the anonymous ITCS reviewers for their helpful feedback. This material is based on work supported by the Vannevar Bush Faculty
Fellowship ONR N00014-23-1-2876, National Science Foundation grants
RI-2312342 and RI-1901403, ARO award W911NF2210266, and NIH award
A240108S001.

\bibliography{main}

\newpage

\newpage

\end{document}

%% file: text/intro.tex
\section{Introduction}
\label{sec:intro}

At the heart of the intricate interplay between online learning and game theory, which can be traced all the way back to Blackwell's seminal approachability theorem~\citep{Blackwell56:analog} and Robinson's analysis of fictitious play~\citep{Robinson51:iterative}, lies the fundamental \emph{no-regret} framework. Here, a learner has to select round by round a sequence of actions so as to obtain a high cumulative reward; the crux presents itself in the online nature of the revealed information, in that each reward function is unbeknownst to the learner prior to the termination of that round. 
The canonical measure of performance in this online environment is the notion of \emph{regret}, which contrasts the cumulative reward of the learner to that of the optimal fixed action in hindsight; a learner is said to incur \emph{no-regret} if its regret grows sublinearly with the time horizon $T$. By now, it is well understood that when the sequence of rewards is produced adversarially, the minimax regret after $T$ repetitions is precisely $\tilde{\Theta}(\sqrt{ T \log m})$, where $m$ denotes the number of available actions of the learner.\footnote{We use the $\tilde{\Theta}(\cdot)$ notation to suppress $\polylog T$ factors.} 
More broadly, online learnability in more general combinatorial domains, such as binary classification, can be characterized by a certain notion of dimension known as the Littlestone dimension~\citep{littlestone1988learning, Ben-David09:Agnostic}, painting a rather complete picture in the adversarial regime. In this context, the alluded nexus between the no-regret framework and game theory can be witnessed by the celebrated realization that players with sublinear regret converge to a certain game-theoretic equilibrium concept known as \emph{coarse correlated equilibrium (CCE)}~\citep{Hart00:Simple,Foster97:Calibrated,Cesa-Bianchi06:Prediction}. In particular, if players follow certain no-regret algorithms, such as the celebrated multiplicative weights update ($\mwu$), the history of play induces an $\epsilon$-CCE after merely $T = O(\frac{\log m}{\epsilon^2})$ repetitions of the game.

In light of our rather comprehensive understanding of online learning, one might expect that the fundamental barriers of no-regret learning in games have already been identified. However, as it turns out, this is not the case. Indeed, the regret incurred by each player when facing other learning agents can be remarkably smaller than what is predicted by the overly pessimistic no-regret framework. This is exemplified by the recent result of~\citet{Daskalakis21:Near}, who proved that when players in an $n$-player $m$-action game follow the \emph{optimistic} counterpart of $\mwu$~\citep{Rakhlin13:Online} (henceforth $\omwu$), each player's regret grows only as $\tilde{O}(n \log m)$, revealing an exponential separation compared to the lower bound in the adversarial regime. Another noteworthy example concerns the behavior of fictitious play: it is hopeless in the adversarial setting, where it can accumulate linear regret~\citep{Cesa-Bianchi06:Prediction}, but Julia Robinson famously proved that it converges to minimax equilibria when followed by both players in a (two-player) zero-sum game~\citep{Robinson51:iterative}.

Despite the considerable interest recent work has devoted to understanding the problem of no-regret learning in games (\Cref{sec:related} features several such results), little is known in terms of lower bounds. \citet{Daskalakis15:Near} made an early effort by noting that incurring $\Omega(1)$ regret is---at least in some sense---inevitable; this boils down to the straightforward realization that even in a single-agent problem the first decision will likely be suboptimal, resulting in $\Omega(1)$ regret even if all the subsequent actions are optimal. Besides failing to provide a meaningful bound in terms of the dimensions of the game, another unsatisfactory feature of the lower bound of~\citet{Daskalakis15:Near} is that it can be bypassed by simply detaching the first iteration---in which case both players actually incur $0$ regret. Can we hope to guarantee that each player will incur $\tilde{O}(1)$ regret, \emph{independent} of the dimensions of the game, or are there fundamental barriers that circumscribe the performance of no-regret learners in games?

\subsection{Our results}

Our primary contribution in this paper is to make a step towards filling the aforementioned knowledge gap by establishing the first non-trivial computational hardness results when multiple players are learning in games.

Our first lower bound concerns a class of no-regret dynamics that includes $\mwu$, perhaps the most well-studied learning algorithm in the literature, as well its optimistic counterpart ($\omwu$). Before we proceed, we recall that the \emph{exponential-time hypothesis (ETH)} for $\PPAD$~\citep{Babichenko16:Can} postulates that there do no exist truly subexponential algorithms for solving $\textsc{EndOfALine}$---the prototypical $\PPAD$-complete problem (\Cref{conj:ETH}). By now, this is a fairly standard computational complexity assumption, which was---crucially for the purpose of this paper---famously invoked by~\citet{Rubinstein18:Inapproximability} to settle the complexity of computing approximate Nash equilibria in two-player (normal-form) games.

\begin{theorem}
    \label{theorem:informalmwu}
    Consider the class of $n$-player $m$-action games in normal form. If each player $i \in \range{n}$ follows $\mwu$ or $\omwu$ and incurs (cumulative) regret $\reg_i^T$, there is a game $\calG$ and an absolute constant $\epsilon > 0$ such that at least $T \geq 2^{(\log_2 \log_2 m)^{1/2 - o(1)}}$ repetitions of the game are needed so that $\frac{1}{T} \max_{1 \leq i \leq n} \reg_i^T \leq \epsilon$, unless ETH for $\PPAD$ (\Cref{conj:ETH}) is false.
\end{theorem}

This represents the first non-trivial lower bounds for no-regret learning in the fundamental setting of~\Cref{theorem:informalmwu} under algorithms such as $\mwu$ and $\omwu$. \Cref{theorem:informalmwu} applies under any choice of learning rates (as specified in \Cref{cor:MWU}). For comparison, we have already alluded to the fact that $T = O(\log m) = O(2^{\log \log m})$ repetitions of $\mwu$ or $\omwu$ suffice to guarantee that $\frac{1}{T} \max_{1 \leq i \leq n} \reg_{i}^T \leq \epsilon$, for any absolute constant $\epsilon > 0$. As such, \Cref{theorem:informalmwu} leaves a certain gap compared to the best known upper bounds. In fact, \Cref{theorem:informalmwu} actually applies to a class of algorithms more general than $\mwu$-type update rules, as we will make clear shortly.

\paragraph{Learning in extensive-form games} Although \Cref{theorem:informalmwu} concerns the behavior of ($\mathtt{O}$)$\mwu$ under the standard normal-form representation of finite games, our approach actually revolves around proving hardness results for \emph{extensive-form} games. The extensive-form representation is typically exponentially more compact---and thereby much more appropriate---when encoding games involving sequential moves as well as imperfect information. In light of their ubiquitous presence in real-world applications, there has been a considerable interest in understanding the performance of no-regret learning algorithms in the more challenging class of extensive-form games (see \Cref{sec:related}). Indeed, no-regret dynamics have been at the heart of recent landmark results in practical computation of strategies for large games~\citep{Bowling15:Heads,Brown17:Superhuman,Brown19:Superhuman,Meta22:Diplomacy}.

\paragraph{Sparse equilibria} To establish lower bounds for no-regret learning algorithms in extensive-form games, we follow a beautiful approach recently put forward by~\citet{foster2023hardness} in a different context, namely that of Markov games. Their idea is to use as a proxy a refinement of CCE in which a certain sparsity constraint is imposed. More precisely, a correlated distribution is said to be \emph{$k$-sparse} if it can be expressed as the uniform mixture of $k$ product distributions (\Cref{def:sparseCCE}); as such, we clarify that a $1$-sparse CCE is equivalent to a Nash equilibrium. The connection of this refinement with the no-regret framework is evident: any CCE derived from (independent) no-regret learners after $T$ repetitions certainly satisfies the $T$-sparsity constraint~\citep{foster2023hardness}. The name of the game now is to establish hardness results for computing sparse CCE in a certain regime of sparsity, which in turn would readily impose barriers on the performance of (computationally efficient) no-regret algorithms.

As an aside, we argue that a $k$-sparse CCE is an important solution concept in its own right, besides the connection with no-regret learning. First, a correlated distribution is in general an exponential object; polynomial sparsity ensures that there exists a succinct representation of that distribution. Indeed, that refinement was central in the celebrated \emph{ellipsoid against hope} algorithm of~\citet{Papadimitriou08:Computing}, the only known polynomial-time algorithm for computing (exact) CCE in succinct multi-player games~\citep{Jiang15:Polynomial}. Further, one important weakness of CCE compared to Nash equilibria is that the former has a much larger description complexity even if the sparsity is polynomial. This becomes especially relevant in some modern machine learning applications in which strategies are represented through massive neural networks, thereby necessitating storing a large sequence of such neural networks in order to simply represent a CCE, which can be prohibitive. In contrast, if a CCE with sparsity $k = 2$ was efficiently computable, that would effectively address such concerns.

\paragraph{Hardness of computing sparse CCE in extensive-form games}

Having motivated the concept of a sparse CCE, we next state our main hardness result for computing such equilibria in extensive-form games under a certain sparsity regime. Below, for an extensive-form game described by a tree $\calT$, we denote by $|\calT|$ the number of nodes in $\calT$ (the reader not familiar with the extensive-form representation can first turn to~\Cref{sec:efgs} for formal definitions).

\begin{restatable}{theorem}{sparsecce}
    \label{theorem:sparseCCE}
    There is no algorithm that runs in time polynomial in the description of an extensive-form game $\calT$ and can compute a $2^{\log_2^{1/2 - o(1)}|\calT|}$-sparse $\epsilon$-CCE, even for an absolute constant $\epsilon > 0$, unless ETH for $\PPAD$ (\Cref{conj:ETH}) is false.
\end{restatable}

Prior to our work, even the complexity status of computing a $2$-sparse $O(1)$-CCE in extensive-form games was open. \Cref{theorem:sparseCCE} implies a superpolynomial separation for the problem of computing sparse CCE between normal- and extensive-form games. Indeed, we have seen that in normal-form games logarithmic---in the description of the game---sparsity is efficiently attainable; \Cref{theorem:sparseCCE} precludes such a possibility in extensive-form games. To better contextualize \Cref{theorem:sparseCCE}, we remark that certain polynomial-time algorithms in extensive-form games attain roughly $2^{\log|\calT|}$-sparsity (in the regime where $\epsilon = O(1)$), thereby leaving again a certain gap compared to \Cref{theorem:sparseCCE}. We further point out that \Cref{theorem:sparseCCE}, and implications thereof (\Cref{theorem:informalmwu} and \Cref{cor:noregret-efgs}), applies even for games with three players; it is open whether it extends to two-player games, a discrepancy explained in more detail in \Cref{sec:main} (\emph{cf.} \citep{Borgs10:Myth,foster2023hardness}).

\paragraph{Implications for no-regret dynamics}

As a consequence, \Cref{theorem:sparseCCE} circumscribes in extensive-form games the performance of any no-regret dynamics that have polynomial complexity per iteration.

\begin{restatable}{corollary}{efgsnoregret}
    \label{cor:noregret-efgs}
    Suppose that each player follows an algorithm with polynomial iteration complexity in the description of an extensive-form game $\calT$. If $\reg_i^T$ is the regret incurred by player $i \in \range{n}$, there is an extensive-form game $\calT$ and an absolute constant $\epsilon > 0$ such that at least $T \geq 2^{\log_2^{1/2 - o(1)}|\calT|}$ repetitions are needed so that $\frac{1}{T} \max_{1 \leq i \leq n} \reg_i^T \leq \epsilon$, unless ETH for $\PPAD$ (\Cref{conj:ETH}) is false.
\end{restatable}

There are many compelling aspects of \Cref{cor:noregret-efgs} worth stressing. First, it applies even in a centralized model well beyond the online and decentralized learning framework. As a concrete example, \Cref{cor:noregret-efgs} applies even if the dynamics are \emph{alternating} instead of simultaneous. Alternation has been a remarkably successful ingredient in practical solvers~\citep{Tammelin15:Solving}, but it violates the online nature of the problem; indeed, the player who gets to play last has complete information about the current reward function. Nevertheless, even alternating dynamics are subject to the barriers imposed by \Cref{cor:noregret-efgs}. Furthermore, as we point out in~\Cref{remark:non-uniform}, \Cref{cor:noregret-efgs} can be extended even if one considers a more general non-uniform notion of regret, which has been observed to lead to significantly faster convergence in practice~\citep{Brown19:Solving}. Another interesting feature of \Cref{cor:noregret-efgs} is that players are allowed to store a polynomial amount of information regarding past rewards. This is considerably stronger---in that designing lower bounds is harder---than the model of~\citet{Daskalakis15:Near} wherein only a constant number of prior rewards can be stored; that assumption was made by~\citet{Daskalakis15:Near} to preclude trivial exploration strategies in two-player zero-sum games whereby players first determine the entire payoff matrix, and then compute a minimax strategy with the information gathered. In contrast, such an exploration strategy is a legitimate possibility in the context of \Cref{cor:noregret-efgs}. As such, the model we consider here is so permissive that no hardness results can be established for no-regret learning in two-player zero-sum games, simply because there are polynomial-time algorithms for computing Nash equilibria in such games.

Returning to \Cref{theorem:informalmwu}, the key connection is that algorithms such as ($\mathtt{O}$)$\mwu$ can be efficiently simulated on the induced normal-form representation of the extensive-form game~\citep{Farina22:Kernelized}. As such, \Cref{theorem:informalmwu} turns out to be a consequence of \Cref{cor:noregret-efgs}. As a result, \Cref{theorem:informalmwu} applies more broadly to any class of algorithms simulated on the induced normal form with per-iteration complexity polynomial in the representation of the underlying extensive-form game.

\paragraph{Technical approach} From a technical standpoint, we follow the approach of~\citet{foster2023hardness}, who proved that computing sparse CCE in Markov (aka. stochastic) games is computationally hard even when targeting a polynomial sparsity. A crucial detail here is that \citet{foster2023hardness} define CCE by allowing potentially non-Markovian deviations, for otherwise polynomial algorithms do exist~\citep{Erez23:Regret}; this already separates regret minimization in Markov games from extensive-form games. The key observation of~\citet{foster2023hardness} is that sparse CCE in general-sum Markov games can be leveraged to efficiently compute \emph{Nash equilibria} in general-sum (normal-form) games, thereby confronting immediate computational barriers~\citep{Daskalakis09:Complexity,Chen09:Settling,Rubinstein18:Inapproximability}. Following this connection, we establish a similar reduction: we show that for any two-player $m$-action game $\calG$ there is an extensive-form game $\calT = \calT(\calG)$ (\Cref{def:lifted-EFG}) with the property that i) a $T$-sparse $\epsilon$-CCE in $\calT$ induces an $O(\epsilon)$-NE in $\calG$ (\Cref{theorem:inference-EFGs}), and ii) the description of $\calT$ is of the order $m^{\log T/\epsilon
^2}$. The key idea is that by repeating the underlying game $\calG$ multiple times, a potentially deviating player could approximately discern the product distribution the rest of the players prescribe to, even though their randomization is unbeknownst to the deviator. In turn, this essentially forces a CCE in $\calT$ to contain a Nash equilibrium strategy for $\calG$ by virtue of a reduction due to~\citet{Borgs10:Myth}; otherwise, there would exist a deviation with a significant profit in $\calT$, contradicting the assumption that the original mixture of product distributions constitutes a CCE. This argument is the crux of the entire approach, and---following~\citet{foster2023hardness}---relies on some classical results on \emph{online density estimation}, namely \emph{Vovk's aggregating algorithm}~\citep{Vovk90:Aggregating}. In particular, Vovk's algorithm guarantees that a deviating player can identify, within $\epsilon$ total variation distance in expectation, the strategy of the rest of the players after $H = O(\frac{\log T}{\epsilon^2})$ repetitions of the game.

\subsection{Further related work}
\label{sec:related}

The line of work endeavoring to characterize the performance of no-regret learners in games, beyond the adversarial regime~\citep{Cesa-Bianchi06:Prediction,Blum07:Learning}, was pioneered by~\citet{Daskalakis15:Near} in the context of two-player zero-sum games. Thereafter, it has attracted considerable interest in the literature~\citep{Farina19:Stable,Piliouras22:Beyond,Rakhlin13:Optimization,Syrgkanis15:Fast,Hsieh21:Adaptive,Hsieh22:No,Chen20:Hedging,Farina21:Faster,Kangarshahi18:Lets,Foster16:Learning,Zhang22:Policy,Yang23:Convergence,Daskalakis22:Fast}, culminating in the breakthrough result of~\citet{Daskalakis21:Near} highlighted earlier in \Cref{sec:intro}.

Yet, despite the significant progress, little is known in terms of lower bounds, with some notable exceptions. First, \citet{Syrgkanis15:Fast} showed that if one player follows $\mwu$ and the other player is best responding in the context of a two-player zero-sum game, one of the players must incur $\Omega(\sqrt{T})$ regret, no matter how the learning rate is set. With a more elaborate argument, \citet{Chen20:Hedging} established the same lower bound when both players follow $\mwu$ in a two-player game, again for any choice of learning rate. Both of those results were constructed based on binary-action games, and as such, they did not provide any meaningful lower bounds in terms of the dimensions of the game. Furthermore, \citet{Hadiji23:Towards} recently investigated the first-order query complexity of computing $\epsilon$-Nash equilibria in $m \times m$ two-player zero-sum games (\emph{cf.} \citep{Goldberg23:Lower,Fearnley16:Finding,Babichenko16:Query,Fearnley15:Learning,Maiti23:Query}). They showed that $\Omega(m)$ (first-order) queries are needed when $\epsilon = 0$, and roughly $\Omega(\log ( \frac{1}{m \epsilon} ) )$ when $\epsilon = O(\frac{1}{m^4})$, thereby leaving a substantial gap with the upper bound of $O(\frac{\log m}{\epsilon})$ attained via $\omwu$. The lower bounds we establish in this paper are quite different, being of computational nature. Indeed, we have already explained that in the more permissive (potentially centralized) model that we study here, there are no obstacles in attaining zero regret in a single iteration of a two-player zero-sum game.

Finally, one of our main results (\Cref{theorem:sparseCCE}) establishes a superpolynomial separation between no-regret learning in extensive- and normal-form games. It is worth stressing thus that learning in extensive-form games has been a particularly popular research topic in the literature (\emph{e.g.},~\citep{Zinkevich07:Regret,Farina21:Better,Kozuno21:Learning,Fiegel23:Adapting,Bai22:Near,Bai22:Efficient,Farina19:Regret,Morrill21:Efficient,Morrill21:Hindsight,Farina19:Stable,Gordon08:No,Heinrich15:Fictitious,Dudik09:SamplingBased,Tang23:Regret,Farina21:Faster,Farina22:Simple,Song22:Sample,Zhang21:Finding,Qi23:Pure,Fiegel23:Local}, and the numerous references therein). This emphasis stems to a large extent from the fact that the extensive-form representation is more suited to capture realistic settings that feature sequential moves and imperfect information.

%% file: text/prels.tex
\section{Preliminaries}
\label{sec:prels}

In this section, we provide the necessary background on normal- and extensive-form games, as well as the setting of online density estimation. Specifically, \Cref{sec:sparseCCE} formalizes the refinement of CCE we focus on; \Cref{sec:efgs} introduces the extensive-form representation; and \Cref{sec:ode} describes Vovk's aggregation algorithm~\citep{Vovk90:Aggregating} in the context of online density estimation. For further background on learning in games, we refer to the excellent book of~\citet{Cesa-Bianchi06:Prediction}.

\paragraph{Conventions} We let $\N = \{1, 2, \dots, \}$ be the set of natural numbers. We oftentimes use the $O(\cdot), \Omega(\cdot), \Theta(\cdot)$ notation with a non-asymptotic semantic so as to suppress absolute constants. For a finite set $S$, we let $\uni(S)$ denote the uniform distribution over $S$. $\log (\cdot)$ denotes the natural logarithmic (with base $e$). We generally use subscripts to indicate the player and superscripts (with parentheses) to specify the (discrete) time index.

%For an $n$-dimensional vector $\va,$ we let $\va_{-i}$ denote the $(n-1)$-dimensional vector $(a_j : j \neq i).$ I added that clarification a bit later

\subsection{Sparse coarse correlated equilibria}
\label{sec:sparseCCE}

As we explained earlier in our introduction, our main focus here is on the problem of computing a refinement of the standard coarse correlated equilibrium (CCE)~\citep{Moulin78:Strategically} satisfying a certain sparsity constraint. To formally introduce that refinement, let us first introduce the normal-form representation of finite games.

\paragraph{Normal-form games} Let $\calG $ be a finite $n$-player game represented in normal form. The set of players will be denoted by $\range{n} \defeq \{1, 2, \dots, n\}$, and will be indexed by variables $i, i' \in \range{n}$. In the normal-form representation, every player $i \in \range{n}$ has a finite and nonempty set of available actions $\calA_i$. For notational convenience, we will let $m \triangleq \max_{1 \leq i \leq n} |\calA_i|$. For every possible combination of actions $\va \defeq (a_1, \dots, a_n) \in \bigtimes_{i=1}^n \calA_i$, there is a utility function $u_i : \bigtimes_{i=1}^n \calA_i \to [-1, 1]$ that specifies the utility (or reward) $u_i(\va)$ of player $i \in \range{n}$ under that joint action; the range of the utilities here can be normalized to be in $[-1, 1]$ without any loss of generality. Each player $i \in \range{n}$ is allowed to randomize by selecting a (mixed) strategy, a distribution over the available actions: $\vx_i \in \Delta(\calA_i) \defeq \{ \vx_i \in \R_{\geq 0}^{\calA_i} : \sum_{a_i \in \calA_i} \vx_i[a_i] = 1 \}$. For a joint strategy $(\vx_1, \dots, \vx_n) \in \bigtimes_{i=1}^n \Delta(\calA_i)$, we will denote by $\bigotimes_{i=1}^n \vx_i$ the \emph{product} distribution on $\Delta(\bigtimes_{i=1}^n \calA_i)$ defined so that $(\bigotimes_{i=1}^n \vx_i) [(a_1, \dots, a_n)] \defeq \prod_{i=1}^n \vx_i[a_i]$.

We are now ready to recall the standard concept of an approximate coarse correlated equilibrium (CCE)~\citep{Moulin78:Strategically,Aumann74:Subjectivity}. Below, for a joint action $\va = (a_1, \dots, a_n) \in \bigtimes_{i=1}^n \calA_i$, we use the usual shorthand notation $\va_{-i} \defeq (a_1, \dots, a_{i-1}, a_{i+1}, \dots, a_n) \in \bigtimes_{i' \neq i} \calA_{i'}$.

\begin{definition}
[Coarse correlated equilibrium]
    \label{def:CCE}
    Let $\calG$ be an $n$-player game in normal form. A distribution over joint action profiles $\vec{\mu} \in \Delta\left( \bigtimes_{i=1}^n \calA_i \right)$ is said to be an \emph{$\epsilon$-coarse correlated equilibrium ($\epsilon$-CCE)}, with $\epsilon \in \R$, if for any player $i \in \range{n}$ and any deviation $a_i' \in \calA_i$,\footnote{As is standard, we abuse notation by parsing $u_i(a_i', \va_{-i})$ as $u_i(a_1, \dots, a_{i-1}, a_i', a_{i+1}, \dots, a_n)$; the same convention is adopted for the mixed extension of the utilities as well.}
    \begin{equation}
        \label{eq:CCE}
        \E_{\vec{a} \sim \vec{\mu}}[u_i(\vec{a})] \geq \E_{\vec{a} \sim \vec{\mu}}[u_i(a_i', \vec{a}_{-i})] - \epsilon.
    \end{equation}
\end{definition}

For convenience, we will sometimes use the shorthand notation $u_i(\vmu) \defeq \E_{\vec{a} \sim \vec{\mu}}[u_i(\vec{a})]$ and $u_i(a_i', \vmu_{-i}) \defeq \E_{\va \sim \vec{\mu}} [u_i(a_i', \va_{-i})]$. A CCE is typically modeled via a trusted third party---a so-called \emph{mediator}---who privately makes recommendations to each player; \eqref{eq:CCE} guarantees that no player can gain more than an additive factor of $\epsilon$ through a (unilateral) deviation, \emph{before} actually observing the mediator's recommendation. A $0$-CCE will simply be referred to as a CCE. In this context, we will be concerned with a refinement of \Cref{def:CCE} wherein a certain sparsity constraint is imposed, in the following formal sense.

\begin{definition}
[Sparse CCE]
    \label{def:sparseCCE}
    Let $\calG$ be an $n$-player game in normal form. A distribution over joint action profiles $\vec{\mu} \in \Delta\left( \bigtimes_{i=1}^n \calA_i \right)$ satisfying \Cref{def:CCE} is said to be \emph{$k$-sparse} if it can be expressed as a uniform mixture of $k$ product distributions; that is, there exist $(\vx^{(1)}_1, \dots, \vx^{(1)}_n), \dots, (\vx_1^{(k)}, \dots, \vx_n^{(k)}) \in \bigtimes_{i=1}^n \Delta(\calA_i)$ such that $\vmu = \frac{1}{k} \sum_{\kappa = 1}^k \bigotimes_{i=1}^n \vx_i^{(\kappa)}$.
\end{definition}

From a computational standpoint, a $\poly(n,m)$-sparse CCE can be identified in polynomial time for any game of \emph{polynomial type}---meaning that $m$ is a polynomial with respect to the underlying description---satisfying the polynomial expectation property~\citep{Papadimitriou08:Computing,Jiang15:Polynomial}; the latter property postulates that for any product distribution $\vx$ with a polynomial representation, the expectation $\E_{\va \sim \vx} [u_i(\va)]$ can be computed in time $\poly(n, m)$, an assumption which is known to be satisfied in most succinct games of interest~\citep{Papadimitriou08:Computing}. On the other end of the spectrum, a $1$-sparse CCE is, by definition, a Nash equilibrium, thereby making the problem $\PPAD$-hard~\citep{Daskalakis09:Complexity}. Furthermore, specific no-regret learning algorithms, such as multiplicative weights update, yield an $O(\frac{\log m}{\epsilon^2})$-sparse $\epsilon$-CCE in time $\poly(n, m, 1/\epsilon)$, again under the polynomial expectation property. As a result, the key question that arises is to characterize the threshold of computational tractability in terms of the sparsity parameter $k$.

\paragraph{Online learning in games} Before we proceed, we also point out the folklore connection between no-regret learning and CCE. In the online learning framework with full feedback, every repetition $t \in \N$ finds a player $i \in \range{n}$ selecting an action $\vx_i^{(t)} \in \Delta(\calA_i)$, and subsequently observing as feedback from the environment the utility function $\vx_i \mapsto \langle \vx_i, \vu_i^{(t)} \rangle  $, where $\vu_i^{(t)} \in [-1, 1]^{\calA_i}$. The \emph{regret} of player $i$ under a time horizon $T \in \N$ is defined as 
\begin{equation*}
    \reg_i^T \defeq \max_{\vx^\star_i \in \Delta(\calA_i)} \sum_{t=1}^T \langle \vx^\star_i - \vx_i^{(t)}, \vu_i^{(t)} \rangle.
\end{equation*}

Specifically, in the setting of learning in games, the utility $\vu_i^{(t)}$ observed by player $i \in \range{n}$ is defined so that
\begin{equation}
    \label{eq:util}
    \vu_i^{(t)}[a_i] \defeq \E_{\va_{-i} \sim \vx^{(t)}_{-i}} [u_i(a_i, \va_{-i})],
\end{equation}
where $\vx^{(t)}_{-i}$ is the joint strategy of the other players at time $t$. In this context, the connection between CCE and no-regret learning is summarized in the following folklore fact~\citep{Blum07:Learning,Cesa-Bianchi06:Prediction}.

\begin{proposition}
    \label{prop:folklore}
    Consider an $n$-player game in normal form, and suppose that each player $i \in \range{n}$ produces the sequence of strategies $(\vx_i^{(t)})$ under the sequence of utilities given by~\eqref{eq:util}. If each player $i $ incurs regret $\reg_i^T$ after $T \in \N$ repetitions, then the correlated distribution $\barvmu \defeq \frac{1}{T} \sum_{t=1}^T \bigotimes_{i=1}^n \vx_i^{(t)}$ is a $\frac{1}{T} \max_{1 \leq i \leq n} \reg_i^T$-CCE.
\end{proposition}

\subsection{Extensive-form games}
\label{sec:efgs}

While every finite game can be represented in normal form, such a representation can be dramatically inefficient in more structured classes of games. Specifically, in scenarios involving sequential moves and imperfect information, the canonical representation is the so-called \emph{extensive form}~\citep{Shoham08:Multiagent}. In such games, a rooted and directed tree $\calT$ is explicitly given as part of the input. We let $\calH = \calH(\calT)$ denote the set of non-leaf nodes of $\calT$. Each node $\tau \in \calH$ that is not a leaf of $\calT$ is uniquely associated with a player $i \in \range{n}$ who selects an action from a finite and nonempty set of available actions $\calA_\tau$;\footnote{In general, extensive-form games also feature \emph{chance moves} (for example, the roll of a dice), which can be modeled via an additional fictitious ``player;'' our lower bounds in the sequel do not have to involve chance moves. Nevertheless, it is worth noting that the addition of chance moves is known to crucially affect the computational complexity of certain problems~\citep{Stengel08:Extensive}.} the set of all nodes where player $i$ acts will be denoted by $\calH_i$. The leaves of the tree $\calT$, which are also referred to as terminal nodes, are denoted by $\calZ$. When the game transitions to a terminal node in $\calZ$, utilities are assigned to each player $i$, as specified by an arbitrary utility function $u_i : \calZ \to [-1, 1]$.

To model imperfect information, the set of nodes $\calH_i$ belonging to player $i$ is partitioned into \emph{information sets} $\calJ_i$; each information set groups together nodes that player $i$ cannot distinguish based on the information structure of the game. As is standard, we also tacitly assume throughout this paper that players have \emph{perfect recall}, in that players never forget acquired information; in the absence of perfect recall, many natural problems immediately become $\NP$-hard~\citep{Koller94:Fast,Tewolde23:Computational,Zhang22:Team,Chu01:NP}. In what follows, we will use $\calT$ to represent the underlying extensive-form game, with the understanding that $\calT$ indeed encodes all the information pertinent for its complete description.

\begin{remark}[Simultaneous moves]
    \label{remark:sim}
    Although the aforedescribed standard formulation of extensive-form games features solely sequential moves, in that each node is associated with a single player, one can readily model simultaneous moves as well through the use of imperfect information. We will use this standard fact in the sequel to also incorporate simultaneous moves in order to simplify the exposition.
\end{remark}

A strategy for a player $i \in \range{n}$ is a mapping $\calJ_i \ni j \to \Delta(\calA_j)$. It turns out that the set of each player's strategies can be represented compactly via a convex polytope~\citep{Stengel96:Efficient,Romanovskii62:Reduction}, namely the \emph{sequence-form polytope} $\calX_i$, so that the utility of each player can be expressed as a linear function (assuming that the rest of the players are fixed). This has been a crucial observation for designing efficient algorithms for a number of fundamental problems in extensive-form games. With a slight abuse of notation, for a sequence-form strategy $\vx_i \in \calX_i$ we will write $\vx_{i, j}$ to denote the probability distribution over $\Delta(\calA_j)$ induced by $\vx_i$ at information set $j \in \calJ_i$; if $\vx_{i, j}$ is not well-defined, in that $\vx_i$ assigns zero probability to the subtree of $j$, we may take $\vx_{i, j}$ to be an arbitrary distribution. For a joint strategy $(\vx_1, \dots, \vx_n) \in \bigtimes_{i=1}^n \calX_i$, we use again the notation $\bigotimes_{i=1}^n \vx_i$ the express the product distribution on the induced normal-form game, which is always represented implicitly.

Extensive-form games are not of polynomial type in that a player's number of pure strategies is typically exponential in the description of the game, rendering the induced normal-form representation largely inefficient. Nevertheless, polynomial (in the size of the tree $\calT$) algorithms for computing (exact) CCE are known to exist. In particular, \citet{Huang08:Computing} have shown how to adapt the algorithm of~\citet{Papadimitriou08:Computing} for certain correlated equilibrium concepts in extensive-form games. (CCE per \Cref{def:sparseCCE} is typically referred to as \emph{normal-form} CCE (NFCCE) to differentiate with other notions of coarse correlation in extensive-form games~\citep{Farina20:Coarse}.) We clarify that a CCE in extensive-form games can be indeed defined via \Cref{def:CCE} through the induced normal-form representation. Our results here revolve around the complexity of the more refined concept introduced in~\Cref{def:sparseCCE}.

\subsection{Online density estimation} 
\label{sec:ode}

We finally conclude this section by recalling some basic results regarding online density estimation, which will be useful in the sequel. Following~\citet{foster2023hardness}, our proof will make use of \emph{Vovk’s aggregating algorithm} for online density estimation~\citep{Vovk90:Aggregating}. More precisely, the setting here is as follows. There are two players, the \emph{nature} and the \emph{learner}. There is also a set $\calO$ called the \emph{outcome space}, and a set $\calC$ referred to as the \emph{context space}; both can be assumed to be finite for our applications. The interaction between the learner and the nature proceeds for $h = 1, 2, \dots, H$ as follows.

\begin{enumerate}
    \item Nature first reveals a context $c_{h} \in \calC$;
    \item the learner then predicts a distribution $\hatvq_{h} \in \Delta(\calO)$ over outcomes based on the observed context $c_{h}$; and
    \item Nature chooses an outcome $o_{h} \in \calO$, and the learner incurs a logarithmic loss defined as $\elllog_{h}(\hatvq_{h}) \defeq \log \left( \frac{1}{\hatvq_{h}(o_{h})} \right)$.
\end{enumerate}

We measure the performance of the learner via the \emph{regret} against a (finite) set of experts $\calE$. In particular, every expert $e \in \calE$ corresponds to a function $\vp^{(e)} : \calC \to \Delta(\calO)$, so that the regret of an algorithm with respect to the expert class $\calE$ is defined as
\begin{equation*}
    \reg_\calE^H \defeq \sum_{h=1}^H \elllog_{h}(\hatvq_{h}) - \min_{e \in \calE} \left\{ \sum_{h=1}^H \elllog_{h}\left(\vec{p}^{(e)}(c_{h}) \right)  \right\}.
\end{equation*}

It is important to note here that the algorithm of the learner has access to the expert predictions $\{ \vec{p}^{(e)}(c_{h}) \}_{e \in \calE}$. In this context, Vovk's aggregating algorithm makes predictions for $h = 1, 2, \dots, H$ via
\begin{equation}
    \label{eq:Vovk}
    \hatvq_{h} \defeq \E_{e \sim \tilvq_{h}}[\vec{p}^{(e)}(c_{h})], \textrm{ where } \tilvq_{h}^{(e)} \defeq \frac{ \exp \left( - \sum_{\upsilon=1}^{h-1} \elllog_{\upsilon}(\vec{p}^{(e)}(c_{\upsilon})) \right) }{ \sum_{e' \in \calE} \exp \left( - \sum_{\upsilon=1}^{h-1} \elllog_{\upsilon}(\vec{p}^{(e')}(c_{\upsilon})) \right) } \quad \forall e \in \calE.
\end{equation}

The convention above is that a summation with no terms is defined as $0$, so that $\tilvq_1$ is the uniform distribution over $\calE$. We further take $\log (\frac{1}{0^+}) = +\infty$ and $\exp(-\infty) = 0$; under realizability (see \Cref{prop:aggregation}), the denominator in~\eqref{eq:Vovk} can never be $0$, so~\eqref{eq:Vovk} is indeed well-defined.

The main guarantee we will use for the aggregation algorithm~\eqref{eq:Vovk} is summarized below. We recall first that the total variation distance between two discrete distributions $\vp, \vq \in \Delta(\calO)$ is defined as $\dtv(\vp, \vq) \defeq \frac{1}{2} \| \vp - \vq\|_1$.

\begin{proposition}[\cite{Vovk90:Aggregating}]
    \label{prop:aggregation}
    Suppose that the distribution of outcomes is realizable under some $e^\star \in \calE$; that is, $o_h \sim \vp^{(e^\star)}(c_h) \mid c_h$ for each $h \in \range{H}$. Then, the predictions $(\hatvq_h)_{1 \leq h \leq H}$ produced by the aggregation algorithm~\eqref{eq:Vovk} satisfy 
    \begin{equation*}
        \frac{1}{H} \sum_{h=1}^H \E \left[\dtv(\hatvq_{h}, \vec{p}^{(e^\star)}(c_{h})) \right] \leq \sqrt{ \frac{\log |\calE|}{H}},
    \end{equation*}
    where the expectation above is with respect to the underlying random process whereby nature selects the sequence of contexts $(c_1, \dots, c_H) \in \calC^H$.
\end{proposition}

%% file: text/conclusion.tex
\section{Conclusions and Open Problems}

In conclusion, we established the first dimension-dependent computational lower bounds for no-regret learning in extensive- and normal-form games, beyond the well-understood adversarial regime in online learning. A number of important questions remain open. Besides the obvious avenue of bridging the gaps between the current upper and lower bounds in extensive-form games, a fundamental question is to understand the complexity of computing sparse CCE (\Cref{def:sparseCCE}) in normal-form games. Indeed, even the complexity status of that problem under a mixture of two product distributions is open, although the fact that algorithms such as $\mwu$ and $\omwu$ require a superconstant number of iterations under any parameterization (\Cref{cor:MWU}) suggests that the problem is hard. Furthermore, our lower bounds apply to \emph{coarse} correlated equilibria; while those naturally translate to stronger equilibrium concepts as well, such as correlated equilibria, it would be interesting to understand whether stronger hardness results can be obtained for such refinements. In particular, in a surprising turn of events, it was recently shown that logarithmic sparsity is possible even for approximate correlated equilibria~\citep{Dagan23:From,Peng23:Fast}; are those guarantees for swap regret tight when learning in games? Interestingly, in a correlated equilibrium players observe additional information, which could potentially speed up the online density estimation procedure.